\documentclass[letterpaper]{article}

\usepackage{times}  
\usepackage{helvet}  
\usepackage{courier}  
\usepackage{url}  
\usepackage{graphicx}  
\usepackage{amssymb,amsmath,amsthm,amsfonts}
\usepackage{mathtools}
\usepackage{temporal}
\usepackage{ifthen,version}
\usepackage{algorithm}
\usepackage[noend]{algpseudocode}
\usepackage{tikz}
\usetikzlibrary{automata,positioning,decorations.markings,arrows,intersections,%
  calc,shapes}

\colorlet{darkgreen}{green!40!black}
\colorlet{darkblue}{blue!60!black}
\colorlet{darkred}{red!50!black}
\colorlet{safecellcolor}{yellow!5}
\colorlet{goodcellcolor}{green!10}
\colorlet{badcellcolor}{blue!10}

\tikzset{
  >=latex,node distance=2cm,on grid,auto, initial text=,
  box state/.style={draw,rectangle,minimum size=8mm,rounded corners},
  prob state/.style={draw,very thick,shape=circle,darkblue,minimum size=3mm,inner sep=0mm},
  every loop/.style={shorten >=0pt},
  accepting state/.style={double distance=1.2pt, outer sep = 0.6pt+\pgflinewidth},
  accepting dot/.style={above=-2.5pt,circle,fill,darkgreen,inner sep=2pt,radius=1pt},
  loop above/.append style={every loop/.append style={out=120, in=60, looseness=6}},
  loop below/.append style={every loop/.append style={out=300, in=240, looseness=6}},
  loop left/.append style={every loop/.append style={out=210, in=150, looseness=6}},
  loop right/.append style={every loop/.append style={out=30, in=330, looseness=6}}
}

{\theoremstyle{definition} \newtheorem{definition}{Definition}}
\newtheorem{theorem}{Theorem}
\newtheorem{lemma}{Lemma}

{\theoremstyle{remark}
  \newtheorem{exmpl}{Example}
  
  }
\newenvironment{example}{\begin{exmpl}}{\qed\end{exmpl}}

\PassOptionsToPackage{usenames,dvipsnames}{xcolor}
\usepackage{wrapfig}

\usepackage{framed}

\usepackage{paralist}

\usepackage{array}
\usepackage{booktabs}
\usepackage[textwidth=1.5cm,textsize=scriptsize,backgroundcolor=blue!10!white]{todonotes}
\usepackage[caption=false]{subfig}
\usepackage{wrapfig}
\usepackage{xcolor}
\usepackage{multirow}

\newcommand{\toolname}{\textsc{Mungojerrie} }

\usepackage{listings}

\usepackage{color}

\definecolor{dkgreen}{rgb}{0,0.6,0}
\definecolor{gray}{rgb}{0.5,0.5,0.5}
\definecolor{mauve}{rgb}{0.58,0,0.82}

\newcommand{\pto}{\xrightharpoondown{}}
\newcommand{\set}[1]{\left\{ #1 \right\}}

\newcommand{\seq}[1]{\langle #1 \rangle}
\newcommand{\FRuns}{\mathit{Runs}_{\mathit{f}}}
\newcommand{\Runs}{\mathit{Runs}}

\newcommand{\EDisct}{\mathit{EDisct}}
\newcommand{\EAvg}{\mathit{EAvg}}
\newcommand{\ECost}{\mathit{Cost}}
\DeclareMathOperator{\PReach}{\mathit{PReach}}
\DeclareMathOperator{\PSat}{\mathit{PSat}}

\newcommand{\Aa}{\mathcal{A}}
\newcommand{\Mm}{\mathcal{M}}
\newcommand{\Gg}{\mathcal{G}}
\newcommand{\Ff}{\mathcal{F}}

\newcommand{\eE}{\mathbb E}

\newcommand{\Real}{\mathbb R}
\newcommand{\Rplus}{\mathbb R_{\geq 0}}
\def\rmdef{\stackrel{\mbox{\rm {\tiny def}}}{=}} 
\newcommand{\DIST}{{\cal D}}
\newcommand{\llim}[2]{\lim_{#1 \uparrow #2}}
\DeclareMathOperator{\supp}{\mathit{supp}}
\DeclareMathOperator{\last}{\mathit{last}}
\DeclareMathOperator{\infi}{\mathsf{inf}}
\DeclareMathOperator{\acc}{\mathsf{Acc}}

 \title{Omega-Regular Objectives in Model-Free
  Reinforcement Learning} \author{Ernst Moritz Hahn\thanks{University
    of Liverpool, UK.} \and Mateo Perez\thanks{University of Colorado
    Boulder, USA.} \and Sven Schewe$^{*}$ \and \and Fabio Somenzi$^{\dag}$
  \and Ashutosh Trivedi$^{\dag}$ \and Dominik Wojtczak$^{*}$}

\begin{document}
\maketitle
\begin{abstract}
  We provide the first solution for model-free reinforcement learning
  of $\omega$-regular objectives for Markov decision processes (MDPs).
  We present a constructive reduction from the almost-sure
  satisfaction of $\omega$-regular objectives to an almost-sure
  reachability problem, and extend this technique to learning how to
  control an unknown model so that the chance of satisfying the
  objective is maximized.  A key feature of our technique is the
  compilation of $\omega$-regular properties into limit-deterministic
  B\"uchi automata instead of the traditional Rabin automata; this
  choice sidesteps difficulties that have marred previous proposals.
  Our approach allows us to apply model-free, off-the-shelf reinforcement
  learning algorithms to compute optimal strategies from the
  observations of the MDP.  We present an experimental evaluation of
  our technique on benchmark learning problems.
\end{abstract}

\section{Introduction} 
\label{sec:introduction}

Reinforcement learning (RL) \cite{Sutton18} is an approach to
sequential decision making in which agents rely on reward signals to
choose actions aimed at achieving prescribed objectives.  Some
objectives, like running a maze, are naturally expressed in terms of
scalar rewards; in other cases the translation is less obvious.  In
this paper we solve the problem of \emph{$\omega$-regular rewards},
that is, the problem of defining scalar rewards for the transitions of
a Markov decision process (MDP) so that strategies to maximize the
probability to satisfy an $\omega$-regular objective may be computed
by off-the-shelf, \emph{model-free} (a.k.a.\ \emph{direct}) RL
algorithms.

Omega-regular languages provide a rich formalism to unambiguously
express qualitative safety and progress requirements of MDPs
\cite{Baier08}.  The most common way to describe an $\omega$-regular
language is via a formula in Linear Time Logic (LTL); other
specification mechanisms include extensions of LTL, various types of
automata, and monadic second-order logic.  A typical requirement that
is naturally expressed as an $\omega$-regular objective prescribes
that the agent should eventually control the MDP to stay within a
given set of states, while at all times avoiding another set of
states.  In LTL this would be written
$(\eventually \always \mathtt{goal}) \wedge (\always \neg
\mathtt{trap})$, where $\mathtt{goal}$ and $\mathtt{trap}$ are labels
attached to the appropriate states, $\eventually$ stands for
``finally,'' and $\always$ stands for ``globally.''

For verification or synthesis, an $\omega$-regular objective is
usually translated into an automaton that monitors the traces of
execution of the MDP \cite{deAlfa98}.  Successful executions cause the automaton to
take certain (accepting) transitions infinitely often, and ultimately
avoid certain (rejecting) transitions.  That is, $\omega$-regular
objectives are about the long-term behavior of an MDP; the frequency
of reward collected is not what matters.  A policy that guarantees no
rejecting transitions and an accepting transition every ten steps, is
better than a policy that promises an accepting transition at each
step, but with probability $0.5$ does not accept at all.

The problem of $\omega$-regular rewards in the context of model-free
RL was first tackled in \cite{Sadigh14} by translating the objective
into a deterministic Rabin automaton and deriving positive and
negative rewards directly from the acceptance condition of the
automaton.  In Section~\ref{sec:overview} we show that their
algorithm, and the extension of \cite{Hiromo15}
may fail to find optimal strategies, and may underestimate the
probability of satisfaction of the objective.


We avoid the problems inherent in the use of deterministic Rabin
automata for model-free RL by resorting to \emph{limit-deterministic}
B\"uchi automata, which, under mild restrictions, were shown by
\cite{Hahn15} to be suitable for both qualitative and quantitative
analysis of MDPs under all $\omega$-regular objectives.  The B\"uchi
acceptance condition, which, unlike the Rabin condition, does not
involve rejecting transitions, allows us to constructively reduce the
almost-sure satisfaction of $\omega$-regular objectives to an
almost-sure reachability problem.  In addition, it is also suitable
for quantitative analysis: the value of a state approximates the
maximum probability of satisfaction of the objective from that state
as a probability parameter approaches $1$.




In this paper we concentrate on model-free approaches and infinitary
behaviors for finite MDPs.  Related problems include model-based RL
\cite{Fu14}, RL for finite-horizon objectives \cite{Li17}, and learning
for efficient verification \cite{Brazdi14}.

This paper is organized as follows.  In Section~\ref{sec:problem} we
introduce definitions and notations.  Section~\ref{sec:overview}
motivates our approach by showing the problems that arise when the
reward of the RL algorithm is derived from the acceptance condition of
a deterministic Rabin automaton.  In Section~\ref{sec:reinforcement}
we prove the results on which our approach is based.  Finally,
Section~\ref{sec:experiment} discusses our experiments.


\section{Preliminaries}
\label{sec:problem}
An $\omega$-\emph{word} $w$ on an alphabet $\Sigma$ is a function
$w : \mathbb{N} \to \Sigma$.  We abbreviate $w(i)$ by $w_i$.  The set
of $\omega$-words on $\Sigma$ is written $\Sigma^\omega$ and a subset
of $\Sigma^\omega$ is an $\omega$-\emph{language} on $\Sigma$.

A \emph{probability distribution} over a finite set $S$ is a function
$d : S {\to} [0, 1]$ such that $\sum_{s \in S} d(s) = 1$. 
Let $\DIST(S)$ denote the set of all discrete distributions over $S$.
We say a distribution ${d \in \DIST(S)}$ is a \emph{point distribution}
if $d(s) {=} 1$ for some $s \in S$.
For a distribution $d \in \DIST(S)$ we write $\supp(d) \rmdef \set{s
  \in S : d(s) > 0}$.

\subsection{Markov Decision Processes}
A {\it Markov decision process} $\Mm$ is a tuple $(S, A, T, AP, L)$ where
$S$ is a finite set of states, $A$ is a finite set of {\it actions}, $T: S \times A \pto
\DIST(S)$ is the {probabilistic transition (partial) function}, $AP$ is the set of {\it
  atomic propositions}, and $L: S \to 2^{AP}$ is the {\it proposition labeling
  function}.

For any state $s \in S$, we let $A(s)$ denote the set of actions available in
the state $s$.  For states $s, s' \in S$ and $a \in A(s)$ as $T(s,
a)(s')$ equals $p(s' | s, a)$.
A {\it run} of $\Mm$ is an $\omega$-word $\seq{s_0, a_1, s_1, \ldots} \in S
\times (A \times S)^\omega$ such that $p(s_{i+1} | s_{i}, a_{i+1}) {>} 0$ for all $i
\geq 0$.
A finite run is a finite such sequence. 
For a {\it run} $r = \seq{s_0, a_1, s_1, \ldots, s_n}$ we define corresponding
labeled run as $L(r) = \set{L(s_0), L(s_1), \ldots} \in (2^{AP})^\omega$.
We write $\Runs^\Mm (\FRuns^\Mm)$  for the set of runs (finite runs) of $\Mm$ 
and $\Runs{}^\Mm(s) (\FRuns{}^\Mm(s))$  for the set of runs (finite runs) of
$\Mm$ starting from state $s$.
For a finite run $r$ we write $\last(r)$ for the last state of the sequence.

A strategy in $\Mm$ is a function $\sigma : \FRuns \to \DIST(A)$ such that
$\supp(\sigma(r)) \subseteq A(\last(r))$.
Let $\Runs^\Mm_\sigma(s)$ denote the subset of runs $\Runs^\Mm_s$ which
correspond to strategy $\sigma$ with the initial state $s$.
Let $\Sigma_\Mm$ be the set of all strategies.
We say that a strategy $\sigma$ is {\it pure} if $\sigma(r)$ is a point
distribution for  all runs $r \in \FRuns^\Mm$ and we say that $\sigma$ is {\it
  stationary} if $\last(r) = \last(r')$ implies $\sigma(r) = \sigma(r')$ for
all runs $r, r' \in \FRuns^\Mm$.  A strategy that is not pure is \emph{mixed}.
A stationary strategy is {\it positional} if it is both pure and stationary.

The behavior of an MDP $\Mm$ under a strategy $\sigma$ is defined on
a probability space
$(\Runs^\Mm_\sigma(s), \Ff_{\Runs^\Mm_\sigma(s)}, \Pr^\sigma_s)$ over
the set of infinite runs of $\sigma$ with starting state $s$.  Given a
real-valued random variable over the set of infinite runs
$f{:}\Runs^\Mm \to \Real$, we denote by $\eE^\sigma_s \set{f}$ the
expectation of $f$ over the runs of $\Mm$ originating at $s$ that
follow strategy $\sigma$.

A \emph{rewardful} MDP is a pair $(\Mm, \rho)$, where $\Mm$ is an MDP
and $\rho: S \times A \to \Real$ is a reward function assigning utility to state-action pairs.
A rewardful MDP $(\Mm, \rho)$ under a strategy $\sigma$ determines a sequence
of random rewards ${\rho(X_{i-1}, Y_i)}_{i \geq 1}$.
Depending upon the problem of interest, there are a number of performance
objectives proposed in the literature:
\begin{enumerate}
\item
  {\it Reachability Probability} $\PReach^T(s, \sigma)$ (with $T \subseteq S$):
  \begin{equation*}
    \Pr{}_s^\sigma \{\seq{s,a_1,s_1,\ldots} \in \Runs_\sigma^\Mm(s) :
    \exists i \scope s_i \in T \}
  \end{equation*}
\item
  {\it Discounted Reward} (with $\lambda \in [0, 1[$\,):
  \[
  \EDisct^\lambda(s, \sigma) =
  \lim_{N \to \infty}  \eE^\sigma_s\Big\{\sum_{1 \leq i \leq N} \lambda^{i-1} \rho(X_{i-1}, Y_i)\Big\}.
  \]
\item
  {\it Average Reward}
  \[
  \EAvg(s, \sigma) =
  \limsup_{N \to \infty}  \frac{1}{N}\eE^\sigma_s\Big\{\sum_{1\leq i
      \leq N} \rho(X_{i-1}, Y_i)\Big\}.
  \]
\end{enumerate}
For an objective $\ECost {\in} \{\PReach^T, \EDisct^\lambda, \EAvg\}$ and an
initial state $s$ we define the optimal cost
$\ECost_*(s)$ as $\sup_{\sigma \in \Sigma_\Mm} \ECost(s, \sigma)$.
A strategy $\sigma$ of $\Mm$ is optimal for the objective $\ECost$ if $\ECost(s,
\sigma) = \ECost_*(s)$ for all $s \in S$.
For an MDP the optimal cost and optimal strategies can be computed in polynomial
time~\cite{Put94}. 


\subsection{$\omega$-Regular Performance Objectives}
A \emph{nondeterministic $\omega$-regular automaton} is a tuple
${\cal A} = (\Sigma,Q,q_0,\delta,\acc)$, where $\Sigma$ is a
finite \emph{alphabet}, $Q$ is a finite set of \emph{states},
$q_0 \in Q$ is the \emph{initial state},
$\delta : Q \times \Sigma \to 2^Q$ is the \emph{transition function},
and $\acc$ is the \emph{acceptance condition}, to be discussed
presently.  A \emph{run} $r$ of ${\cal A}$ on $w \in \Sigma^\omega$ is
an $\omega$-word $r_0, w_0, r_1, w_1, \ldots$ in
$(Q \cup \Sigma)^\omega$ such that $r_0 = q_0$ and, for $i > 0$,
$r_i \in \delta(r_{i-1},w_{i-1})$.  Each triple
$(r_{i-1},w_{i-1},r_i)$ is a \emph{transition} of ${\cal A}$.

We consider two types of acceptance conditions---B\"uchi and Rabin---that depend
on the transitions that occur infinitely often in a run of an automaton.
A \emph{B\"uchi} (Rabin) automaton is an $\omega$-automaton equipped with a
B\"uchi (Rabin) acceptance condition. 


We write $\infi(r)$ for the set of transitions that appear infinitely
often in the run $r$.
The \emph{B\"uchi} acceptance condition
defined by $F \subseteq Q \times \Sigma \times Q$ is the set of runs
$\{ r \in (Q \cup \Sigma)^\omega : \infi(r) \cap F \neq \emptyset\}$.
A \emph{Rabin} acceptance condition is defined in terms of $k$ pairs
of subsets of $Q \times \Sigma \times Q$,
$(B_0,G_0), \ldots, (B_{k-1},G_{k-1})$, as the set
$\{ r \in (Q \cup \Sigma)^\omega : \exists i < k \scope \infi(r)
\cap B_i = \emptyset \wedge \infi(r) \cap G_i \neq \emptyset \}$.  The
\emph{index} of a Rabin condition is its number of pairs.


A run $r$ of ${\cal A}$ is \emph{accepting} if $r \in \acc$.  The
\emph{language} of ${\cal A}$ (or, \emph{accepted} by ${\cal A}$) is
the subset of words in $\Sigma^\omega$ that have accepting runs in
${\cal A}$.  A language is $\omega$-\emph{regular} if it is accepted
by an $\omega$-regular automaton.

Given an MDP $\Mm$ and an $\omega$-regular objective $\varphi$ given as
an $\omega$-regular automaton $\Aa_\varphi = (\Sigma,Q,q_0,\delta,\acc)$,
we are interested in computing an optimal strategy satisfying the
objective.  We define the satisfaction probability of a strategy
$\sigma$ from initial state $s$ as:
\begin{equation*}
\PSat^\varphi(s, \sigma) =   \Pr{}_s^\sigma \{ r \in \Runs_\sigma^\Mm(s) :
  L(r) \in \acc \}
\end{equation*}
The optimal probability of satisfaction $\PSat_*$ and corresponding optimal
strategy is defined in a manner analogous to other performance objectives.

\subsection{Deterministic Rabin and B\"uchi Automata}

An automaton ${\cal A} = (\Sigma,Q,q_0,\delta,\acc)$ is
\emph{deterministic} if $|\delta(q,\sigma)| \leq 1$ for all $q \in Q$
and all $\sigma \in \Sigma$.  ${\cal A}$ is \emph{complete} if
$|\delta(q,\sigma)| \geq 1$.  A word in $\Sigma^\omega$ has exactly
one run in a deterministic, complete automaton.
We use common three-letter abbreviations to distinguish types of automata.  The
first (D or N) tells whether the automaton is deterministic; the second
denotes the acceptance condition (B for B\"uchi and R for Rabin).
The third letter refers to the type of objects read by the automaton (here, we
only use W for $\omega$-words). For example,
an NBW is a nondeterministic B\"uchi automaton, and a DRW is a
deterministic Rabin automaton.


Every $\omega$-regular language is accepted by some DRW and by some
NBW.  In contrast, there are $\omega$-regular languages that are not
accepted by any DBW.  The \emph{Rabin index} of a Rabin automaton is
the index of its acceptance condition.  The Rabin index of an
$\omega$-regular language $\mathcal{L}$ is the minimum index among
those of the DRWs that accept $\mathcal{L}$.  For each
$n \in \mathbb{N}$ there exist $\omega$-regular languages of Rabin
index $n$.  The languages accepted by DBWs, however, form a
proper subset of the languages of index $1$.

Given an MDP $\Mm = (S, A, T, AP, L)$ with a designated initial state
$s_0 \in S$, and a deterministic $\omega$-automaton
$\mathcal{A} = (2^{AP}, Q, q_0, \delta, \acc)$, the \emph{product}
$\Mm \times \mathcal{A}$ is the tuple
$(S \times Q, (s_0,q_0), A, T^\times, \acc^\times)$.  The
probabilistic transition function
$T^\times : (S \times Q) \times A \pto \DIST(S \times Q)$ is defined
by
\begin{equation*}
  T^\times((s,q),a)((\hat{s},\hat{q})) =
  \begin{cases}
    T(s,a)(\hat{s}) & \text{if } \{\hat{q}\} = \delta(q,L(s)) \\
    0 & \text{otherwise.}
  \end{cases}
\end{equation*}
If $\mathcal{A}$ is a DBW, $\acc$ is defined by
$F \subseteq Q \times 2^{AP} \times Q$; then
$F^\times \subseteq (S \times Q) \times A \times (S \times Q)$ defines
$\acc^\times$ as follows: $((s,q),a,(s',q')) \in F^\times$ if and only
if $(q,L(s),q') \in F$ and $T(s,a)(s') \neq 0$.  If $\mathcal{A}$ is a
DRW of index $k$,
$\acc^\times = \{(B^\times _0,G^\times_0), \ldots,
(B^\times_{k-1},G^\times_{k-1})\}$.  To set $B_i$ of $\acc$, there
corresponds $B^\times_i$ of $\acc^\times$ such that
$((s,q),a,(s',q')) \in B^\times_i$ if and only if
$(q,L(s),q') \in B_i$ and $T(s,a)(s') \neq 0$.  Likewise for
$G^\times_i$.

Given an MDP $\Mm = (S, A, T, AP, L)$, we define its underlying graph as the
directed graph $\Gg_\Mm = (V, E)$ where $V = S$ and $E \subseteq S \times S$ is
such that $(s, s') \in E$ if $T(s, a)(s') > 0$ for some $a\in A(s)$.
A sub-MDP of $\Mm$ is an MDP $\Mm' = (S', A', T', AP, L')$ where $S' \subset S$,
$A' \subseteq A$ such that $A'(s) \subseteq A(s)$ for every $s \in S'$, and $T'$
and $L'$ are analogous to $T$ and $L$ when restricted to $S'$ and $A'$.
In particular, $\Mm'$ is closed under probabilistic transitions, i.e. for all $s
\in S'$ and $a \in A'$ we have that $T(s, a)(s') > 0$ implies that $s' \in S'$.
An {\it end-component}~\cite{deAlfa98} of an MDP $\Mm$ is a sub-MDP $\Mm'$ of
$\Mm$ such that $\Gg_{\Mm'}$ is strongly connected.
A maximal end-component is an end-component that is maximal under
set-inclusion.
Every state $s$ of an MDP $\Mm$ belongs to at most one maximal end-component
of $\Mm$.

\begin{theorem}[End-Component Properties~\cite{deAlfa98}]
  \label{th:end-comp}
  Once an end-component $C$ of an MDP is entered, there is a strategy
  that visits every state-action combination in $C$ with probability
  $1$ and stays in $C$ forever.  Moreover, for every strategy an
  end-component is visited with probability $1$.
\end{theorem}

End components and runs are defined for products just like for MDPs.
A run of $\Mm \times \mathcal{A}$ is accepting if it satisfies the
product's acceptance condition. An \emph{accepting end component} of
$\Mm \times \mathcal{A}$ is an end component such that every run of
the product MDP that eventually dwells in it is accepting.

In view of Theorem~\ref{th:end-comp}, satisfaction of an
$\omega$-regular objective $\varphi$ by an MDP $\Mm$ can be formulated
in terms of the accepting end components of the product
$\Mm \times \mathcal{A}_\varphi$, where $\mathcal{A}_\varphi$ is an
automaton accepting $\varphi$.  The maximum probability of
satisfaction of $\varphi$ by $\Mm$ is the maximum probability, over
all strategies, that a run of the product
$\Mm \times \mathcal{A}_\varphi$ eventually dwells in one of its
accepting end components.

It is customary to use DRWs instead of DBWs in the construction of the
product, because the latter cannot express all $\omega$-regular
objectives.  On the other hand, general NBWs are not used since
causal strategies cannot optimally resolve nondeterministic choices
because that requires access to future events \cite{Vardi85}.

\subsection{Limit-Deterministic B\"uchi Automata}
\label{sec:ldbw}

In spite of the large gap between DRWs and DBWs in terms of indices,
even a very restricted form of nondeterminism is sufficient to make
DBWs as expressive as DRWs.  A \emph{limit-deterministic} B\"uchi
automaton (LDBW) is an NBW
${\cal A} = (\Sigma, Q_i \cup Q_f,q_0,\delta,F)$ such that
\begin{itemize}
\item $Q_i \cap Q_f = \emptyset$;
\item $F \subseteq Q_f \times \Sigma \times Q_f$;
\item $|\delta(s,\sigma)| \leq 1$ for all $q \in Q_f$ and
  $\sigma \in \Sigma$; 
\item $\delta(q,\sigma) \subseteq Q_f$ for all $q \in Q_f$ and
  $\sigma \in \Sigma$.
\end{itemize}
Broadly speaking, a LDBW behaves deterministically once it has seen an accepting
transition. 


LDBWs are as expressive as general NBWs.  Moreover, NBWs can be
translated into LDBWs that can be used for the qualitative and
quantitative analysis of MDPs
\cite{Vardi85,Courco95,Hahn15,Sicker16b}.  We use the translation from
\cite{Hahn15}, which uses LDBWs that consist of two parts: an initial
deterministic automaton (without accepting transitions) obtained by a
subset construction; and a final part produced by a breakpoint
construction.  They are connected by a single ``guess'', where the
automaton guesses a subset of the reachable states to start the
breakpoint construction.  Like other constructions (e.g.\
\cite{Sicker16b}), one can compose the resulting automata with an MDP,
such that the optimal control of the product defines a control on the
MDP that maximises the probability of obtaining a word from the
language of the automaton.  We refer to LDBWs with this property as
\emph{suitable} limit-deterministic automata (SLDBWs) (cf.\ Theorem
\ref{th:recall} for details).


\subsection{Linear Time Logic Objectives}
LTL (Linear Time Logic) is a temporal logic whose formulae describe a
subset of the $\omega$-regular languages, which is often used to
specify objectives in human-readable form.  Translations exist from
LTL to various forms of automata, including NBW, DRW, and SLDBW.  Given
a set of atomic propositions $AP$, $a$ is an LTL formula for each
$a \in AP$.  Moreover, if $\varphi$ and $\psi$ are LTL formulae, so
are
\begin{equation*}
  \neg \varphi, \varphi \vee \psi, \nextt \varphi, \psi \until \varphi \enspace.
\end{equation*}
Additional operators are defined as abbreviations:
$\top \rmdef a \vee \neg a$; $\bot \rmdef \neg \top$;
$\varphi \wedge \psi \rmdef \neg (\neg \varphi \vee \neg \psi)$;
$\varphi \rightarrow \psi \rmdef \neg \varphi \vee \psi$;
$\eventually \varphi \rmdef \top \until \varphi$; and
$\always \varphi \rmdef \neg \eventually \neg \varphi$.  We write $w
\models \varphi$ if
$\omega$-word $w$ over $2^{AP}$ satisfies LTL formula $\varphi$. 
The satisfaction relation is defined inductively.  Let $w^j$ be the
$\omega$-word defined by $w^j_i = w_{i+j}$; let $a$ be an atomic
proposition; and let $\varphi$, $\psi$ be LTL formulae.  Then
\begin{itemize}
\item $w \models a$ if and only if $a \in w_0$;
\item $w \models \neg \varphi$ if and only if $w \not\models \varphi$;
\item $w \models \varphi \vee \psi$ if $w \models \varphi$ or $w
  \models \psi$;
\item $w \models \nextt \varphi$ if and only if $w^1 \models \varphi$;
\item $w \models \psi \until \varphi$ if and only if there exists
  $i \geq 0$ such that $w^i \models \varphi$ and, for $0 \leq j < i$,
  $w^j \models \psi$.
\end{itemize}
Further details may be found in \cite{Manna91}.

\subsection{Reinforcement Learning}
For a given MDP and a performance objective
$\ECost {\in} \{\PReach^T, \EDisct^\lambda, \EAvg\}$, the optimal cost
and an optimal strategy can be computed in polynomial time using value
iteration, policy iteration, or linear programming~\cite{Put94}.  On
the other hand, for $\omega$-regular objectives (given as DRW, SLDBW,
or LTL formulae) optimal satisfaction probabilities and strategies can
be computed using graph-theoretic techniques (computing accepting
end-component and then maximizing the probability to reach states in
such components) over the product structure.  However, when the MDP
transition/reward structure is unknown, such techniques are not
applicable.

For MDPs with unknown transition/reward structure, {\it reinforcement learning}~\cite{Sutton18}
provides a framework to learn optimal strategies from repeated
interactions with the environment. 
There are two main approaches to reinforcement learning in MDPs:
{\it model-free/direct approaches} and {\it model-based/indirect
  approaches}.
In a model-based approach, the learner interacts with the system to first
estimate the transition probabilities and corresponding rewards, and then uses
standard MDP algorithms to compute the optimal cost and strategies.
On the other hand in a model-free approach---such as Q-learning---the learner
computes optimal strategies without explicitly estimating the transition
probabilities and rewards.
We focus on model-free RL to learn a strategy that
maximizes the probability of satisfying a given $\omega$-regular objective.


\section{Problem Statement and Motivation}
\label{sec:overview}

The problem we address in this paper is the following:

\begin{quote}
\emph{  Given an MDP $\Mm$ with unknown transition structure and an
  $\omega$-regular objective $\varphi$, compute a strategy
  that maximizes the probability of $\Mm$ satisfying $\varphi$.}
\end{quote}
To apply model-free RL algorithms to this task, one needs to define
rewards that depend on the observations of the MDP and reflect the
satisfaction of the objective.  It is natural to use the product of
the MDP and an automaton monitoring the satisfaction of the objective
to assign suitable rewards to various actions chosen by the learning
algorithm.

Sadigh \emph{et al.\ }\cite{Sadigh14} were the first ones to consider model-free RL
to solve a qualitative-version of this problem, i.e., to learn a strategy
that satisfies the property with probability $1$.
For an MDP $\Mm$ and a DRW $\Aa_\varphi$ of index $k$, they formed
the product MDP $\Mm \times \Aa_\varphi$ with $k$ different ``Rabin'' reward functions
$\rho_1, \ldots, \rho_k$.
The function $\rho_i$ corresponds to the Rabin pair $(B^\times_i,
G^\times_i)$ and is defined such that for for $R_+, R_- \in \Rplus$: 
\[
\rho_i(e) = \begin{cases}
  -R_{-} & \text{if $e \in B^\times_i$}\\
  R_{+} & \text{if $e \in G^\times_i$}\\
  0 & \text{otherwise.}\\
\end{cases} 
\]
\cite{Sadigh14} claimed that if there exists a strategy satisfying an
$\omega$-regular objective $\varphi$ with probability $1$, then there
exists a Rabin pair $i$, discount factor $\lambda_* \in [0, 1[$, and
suitably high ratio $R_*$, such that for all
$\lambda \in [\lambda_*, 1[$ and $R_{-}/R_{+} \geq R_*$, any strategy
maximizing $\lambda$-discounted reward for the MDP
$(\Mm {\times} \Aa_\varphi, \rho_i)$ also satisfies the
$\omega$-regular objective $\varphi$ with probability $1$.  Using
Blackwell-optimality theorem~\cite{Hordijk2002}, a paraphrase of this
claim is that if there exists a strategy satisfying an
$\omega$-regular objective $\varphi$ with probability $1$, then there
exists a Rabin pair $i$ and suitably high ratio $R_*$, such that for
all $R_{-}/R_{+} \geq R_*$, any strategy maximizing expected average
reward for the MDP $(\Mm {\times} \Aa_\varphi, \rho_i)$ also satisfies
the $\omega$-regular objective $\varphi$ with probability $1$.  There
are two flaws with this approach.
    \begin{enumerate}
    \item
      We provide in Example~\ref{ex:two-pairs} an MDP and an
      $\omega$-regular objective $\varphi$ with 
      Rabin index $2$, such that, although there is a strategy that satisfies the
      property with probability $1$, optimal average strategies corresponding
      to any Rabin reward do not satisfy the objective with probability $1$.
    \item Even for an $\omega$-regular objective with one Rabin pair
      $(B, G)$ and $B {=} \emptyset$---i.e., one that can be specified
      by a DBW---we demonstrate in Example~\ref{ex:average-reward}
      that the problem of finding a strategy that satisfies the
      property with probability $1$ may not be reduced to finding
      optimal expected average strategies.
    \end{enumerate}
  \begin{figure}
  \centering
  \begin{tikzpicture}
    [label name/.style={circle,fill=white,inner sep=1pt}]
    \path[fill=safecellcolor] (0,0) rectangle ++(1,1);
    \path[fill=goodcellcolor] (0,1) rectangle ++(1,1);
    \path[fill=goodcellcolor] (1,0) rectangle ++(1,1);
    \path[fill=badcellcolor]  (1,1) rectangle ++(1,1);
    \draw[thick] (0,0) rectangle ++(2,2);
    \draw (0,1) -- (2,1) (1,0) -- (1,2);
    \node (G0) at (0.5, 1.5) {$g_0$};
    \node (G1) at (1.5, 0.5) {$g_1$};
    \node (B)  at (1.5, 1.5) {$b$};
    \node (R0) at (0, 0.5) [label={left:$0$}] {};
    \node (R1) at (0, 1.5) [label={left:$1$}] {};
    \node (C0) at (0.5, 0) [label={below:$0$}] {};
    \node (C1) at (1.5, 0) [label={below:$1$}] {};
  \end{tikzpicture}
  \vspace{3mm}

  \begin{minipage}{0.48\linewidth}
  \begin{tikzpicture}[scale=0.95,transform shape]
    \node[state,initial,fill=safecellcolor] (Safe) {$\mathtt{safe}$};
    \node[state,fill=badcellcolor] (Trap) [right=2.75cm of Safe]
    {$\mathtt{trap}$};
    \path[->]
    (Safe) edge node {$t_4$} (Trap)
    (Trap) edge[loop right] node {$t_5$} ();
    \path[->]
    (Safe) edge[loop,out=10,in=80,looseness=4] node[midway,swap] {$t_3$} (Safe)
    edge[loop,out=100,in=170,looseness=4] node[midway,swap] {$t_0$} (Safe)
    edge[loop,out=190,in=260,looseness=4] node[midway,swap] {$t_1$} (Safe)
    edge[loop,out=280,in=350,looseness=4] node[midway,swap] {$t_2$} (Safe);
  \end{tikzpicture}
  \end{minipage}
  \begin{minipage}{0.5\linewidth}
    \scalebox{0.9}{
  \begin{tabular}[c]{ccc}
    transition & label & Rabin sets \\\hline
    $t_0$ & $\neg b \wedge \neg g_0 \wedge \neg g_1$ & $B_0, B_1$ \\
    $t_1$ & $\neg b \wedge \neg g_0 \wedge g_1$ & $B_0, G_1$ \\
    $t_2$ & $\neg b \wedge g_0 \wedge \neg g_1$ & $G_0, B_1$ \\
    $t_3$ & $\neg b \wedge g_0 \wedge g_1$ & $G_0, G_1$ \\
    $t_4$ & $b$ & \\
    $t_5$ & $\top$ & 
  \end{tabular}
  }
  \end{minipage}

    \begin{tikzpicture}
    \let\scriptstyle\textstyle
    \node[box state,initial=left,fill=safecellcolor] (S0safe)
         {$\substack{(0,0)\\ \mathtt{safe}}$};
         
    \node[prob state,fill=safecellcolor] (S0prob) [below=2cm of S0safe] {};

    \node[box state,fill=goodcellcolor] (S1safe)
    [right=2.2cm of S0prob] {$\substack{(1,0)\\\mathtt{safe}}$};

    \node[prob state,fill=safecellcolor] (S1prob) [right=2.3cm of S1safe] {};

    \node[box state,fill=goodcellcolor] (S2safe)
    [left=2.2cm of S0prob] {$\substack{(0, 1)\\\mathtt{safe}}$};

    \node[prob state,fill=safecellcolor] (S2prob) [left=2.3cm of S2safe] {};

    \node[draw,shape=ellipse, fill=badcellcolor] (Trap)
    [below=4cm of S0safe] {$\substack{* \\\mathtt{trap}}$};

    \path[->]
    (S0safe) edge[-,swap] node {$\substack{\mathtt{go}\\\{B_0, B_1\}}$} (S0prob)
    (S0prob) edge[swap] node {$p$} (S2safe) edge node {$1{-}p$} (S1safe)
    (S1safe) edge[-,swap] node[above] {$\substack{\mathtt{go}\\\{G_0, B_1\}}$} (S1prob)
    (S1prob) edge[swap, out=90, in=0] node {$1{-}p$} (S0safe) edge[out=270, in=0] node {$p$} (Trap)
    (S2safe) edge[-,swap] node {$\substack{\mathtt{go}\\ \{B_0, G_1\}}$} (S2prob)
    (S2prob) edge[swap, out=90, in=180] node[above] {$p$} (S0safe) edge[out=270,
    in=180] node[below] {$1{-}p~~~~$} (Trap)
    
    (Trap)   edge [loop above, looseness=6] node {$\mathtt{go, rest}$} (Trap)
    (S1safe)   edge [loop above, looseness=6] node
      {$\substack{\mathtt{rest}\\\{G_0, B_1\}}$} (S1safe)
    (S2safe)   edge [loop above, looseness=6] node {$\substack{\mathtt{rest}\\ \{B_0, G_1\}}$} (S2safe)
    (S0safe)   edge [loop above, looseness=6] node {$\substack{\mathtt{rest}\\ \{B_0, B_1\}}$} (S0safe);
    \end{tikzpicture}
  \caption{ A grid-world example (top), a Rabin automaton for $[(\eventually\always g_0) \vee
      (\eventually\always g_1)] \wedge \always \neg b$ (center), and
 product MDP (bottom).}
  \label{fig:two-rabin-pairs}
  \end{figure}
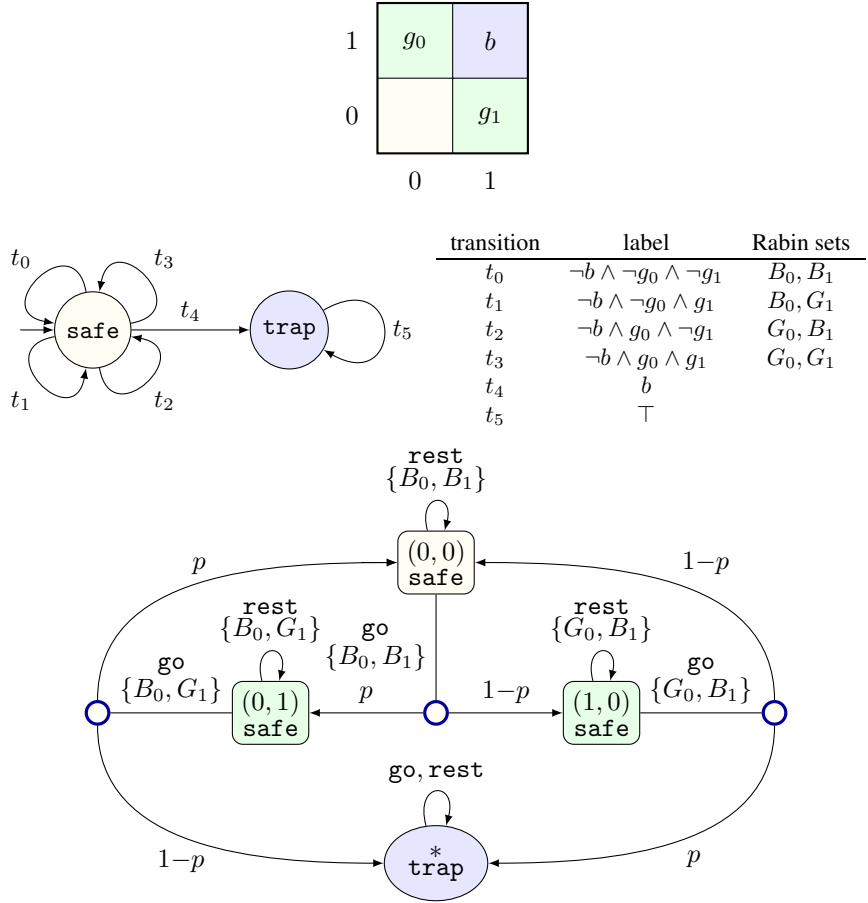

    \begin{example}[Two Rabin Pairs]
  \label{ex:two-pairs}
  Consider the MDP given as a simple grid-world example shown in
  Figure~\ref{fig:two-rabin-pairs}. 
  Each cell (state) of the MDP is labeled with the atomic propositions that are
  true there. 
  In each cell, there is a choice between two actions $\mathtt{rest}$ and
  $\mathtt{go}$.
  With action $\mathtt{rest}$ the state of the MDP does not change. However,
  with action $\mathtt{go}$ the MDP moves to the other cell in the same row
  with probability $p$, or to the other cell in the same column with probability
  $1{-}p$. 
  The initial cell is $(0,0)$.

The specification is given by LTL formula
\[
\varphi = [(\eventually\always g_0) \vee (\eventually\always g_1)]
\wedge \always \neg b.
\]
A DRW that accepts $\varphi$ is shown in
Figure~\ref{fig:two-rabin-pairs}.  The DRW has two accepting pairs:
$(B_0,G_0)$ and $(B_1,G_1)$.  The table besides the automaton gives, for
each transition, its label and the $B$ and $G$ sets to which it
belongs.

The optimal strategy that satisfies the objective $\varphi$ with probability $1$ is
to choose $\mathtt{go}$ in the first Cell~$(0,0)$ and to choose $\mathtt{rest}$
subsequently no matter which state is reached.
However, notice that for both Rabin pairs, the optimal strategy for expected average
reward (or analogously optimal discounted strategies for sufficiently large
discount factors $\lambda$)  is to maximize the probability of reaching one of
the $(0, 1), \texttt{safe}$ or $(1, 0), \texttt{safe}$ states of the product and
stay there forever. 
For the first accepting pair the maximum probability of satisfaction is
$\frac{1}{2-p}$, while for the second pair it is $\frac{1}{1+p}$.
\end{example}

\begin{figure}
  \centering
  \begin{minipage}{0.45\textwidth}
      \begin{tikzpicture}[scale=1.25,action/.style={->,very thick,darkred},
          cell name/.style={circle,fill=white,inner sep=1pt}]
        \coordinate (t1) at (0,0);
        \coordinate (t2) at ($ (t1) + (2,0) $);
        \coordinate (t3) at ($ (t1) + (60:2) $);
        \coordinate (c0) at ($ (t1) + (2,{2/sqrt(3)}) $);
        \foreach \x in {1, 2, 3} {
          \pgfmathparse{\x == 1 ? "badcellcolor" : "goodcellcolor"}
          \edef\cellcolor{\pgfmathresult}
          \path[fill=\cellcolor] (t\x) -- ++(2,0) -- ++(120:2) --cycle;
          \coordinate (c\x) at ($ (t\x) + (1,{1/sqrt(3)}) $);
        }
        \draw[thick] (0,0) -- ++(4,0) -- ++(120:4) --cycle;
        \draw[fill=safecellcolor] (2,0) -- ++(60:2) -- ++(-2,0) --cycle;
        \node[cell name] at ($ (t1) + (2,3) $) {$3$};
        \node[cell name] at ($ (t2) + (-0.6,1.5) $) {$0$};
        \node[cell name] at ($ (t1) + (0.4,0.25) $) {$1$};
        \node[cell name] at ($ (t2) + (1,1.3) $) {$2$};
        \draw[action] ($ (c0) + (0.1,0) $) -- node[pos=0.25,right] {$a$} +(0,0.9);
        \draw[action] ($ (c0) - (0,0.2) $) -- node[pos=0.4,right]
             {$b$} ++(0,-0.5) -- node[at end,below,black] {$1-p$} ++(-0.7,0);
             \draw[action] ($ (c0) - (0,0.2) $) ++(0,-0.5) --
             node[at end,below,black] {$p$} ++(0.7,0);
             \draw[action] ($ (c3) - (0.1,0) $) -- node[pos=0.25,left] {$c$} +(0,-0.9);
             \draw[action] (c1) -- node[pos=0.1,above] {$d$} ++(30:0.9);
             \draw[action] ($ (c2) + (-30:0.2) $)
             arc[radius=0.4,start angle=150,end angle=490] node[at start,right] {$e$};
             \draw[action] (c2) -- node[pos=0.1,above] {$f$} ++(150:0.9);
      \end{tikzpicture}
  \end{minipage}
  \begin{minipage}{0.45\textwidth}
      \begin{tikzpicture}
        \node[state,initial,fill=safecellcolor] (G) {$\mathtt{safe}$};
        \node[state,fill=badcellcolor] (B) [right=2cm of G]
        {$\mathtt{trap}$};
        \path[->]
        (G) edge [loop above] node[accepting dot,label={$\mathtt{g} \wedge \neg \mathtt{b}$}] {} ()
        edge [loop below] node {$\neg\mathtt{g} \wedge \neg \mathtt{b}$} ()
        edge              node {$\mathtt{b}$} (B)
        (B) edge [loop below] node {$\top$} ();
      \end{tikzpicture}
  \end{minipage}
  \vspace{3mm}
  
  \begin{minipage}{0.8\textwidth}
    \centering
      \begin{tikzpicture}
        \node[box state,initial above,fill=safecellcolor] (S0safe)
        {$0,\mathtt{safe}$};
        \node[box state,fill=safecellcolor] (S3safe)
        [below right=1.5cm and 1.5cm of S0safe] {$3,\mathtt{safe}$};
        \node[prob state,fill=safecellcolor] (S0prob) [below left=1.5cm and 1.5cm of S0safe] {};
        \node[box state,fill=safecellcolor] (S1safe)
        [below right=1.5cm and 1cm of S0prob] {$1,\mathtt{safe}$};
        \node[box state,fill=safecellcolor] (S2safe)
        [below left=1.5cm and 1cm of S0prob] {$2,\mathtt{safe}$};
        \node[draw,shape=ellipse, fill=badcellcolor] (Trap)
        [right=2.2cm of S1safe] {$*,\mathtt{trap}$};
        \path[->]
        (S0safe) edge node {\textcolor{darkred}{$a$}} (S3safe)
        edge[-,swap] node {\textcolor{darkred}{$b$}} (S0prob)
        (S0prob) edge[swap] node {$p$} (S2safe)
        edge node {$1-p$} (S1safe)
        (S1safe) edge node {\textcolor{darkred}{$d$}} (Trap)
        (Trap)   edge [loop right, looseness=5] node {} (Trap)
        (S2safe) edge [loop left,looseness=4,out=200,in=160] node[accepting dot,
        label={left:\textcolor{darkred}{$e$}}] {} (S2safe)
        edge[bend left=50] node[accepting dot,
        label={above left:\textcolor{darkred}{$f$}}] {} (S0safe)
        (S3safe) edge[bend right=50] node[accepting dot,
        label={above right:\textcolor{darkred}{$c$}}] {} (S0safe);
      \end{tikzpicture}
  \end{minipage}
  \caption{A grid-world example.  The arrows represent actions (top left).  When
    action $b$ is performed, Cell~2 is reached with probability $p$ and
    Cell~$3$ is reached with probability $1-p$, for $0 < p < 1$. Automaton for
    $\varphi = (\always \neg \mathtt{b}) \wedge (\always\eventually 
    \mathtt{g})$ (top right). The dotted transition is the only accepting transition.
    Product MDP (bottom).}
  \label{fig:grid-world}
\end{figure}
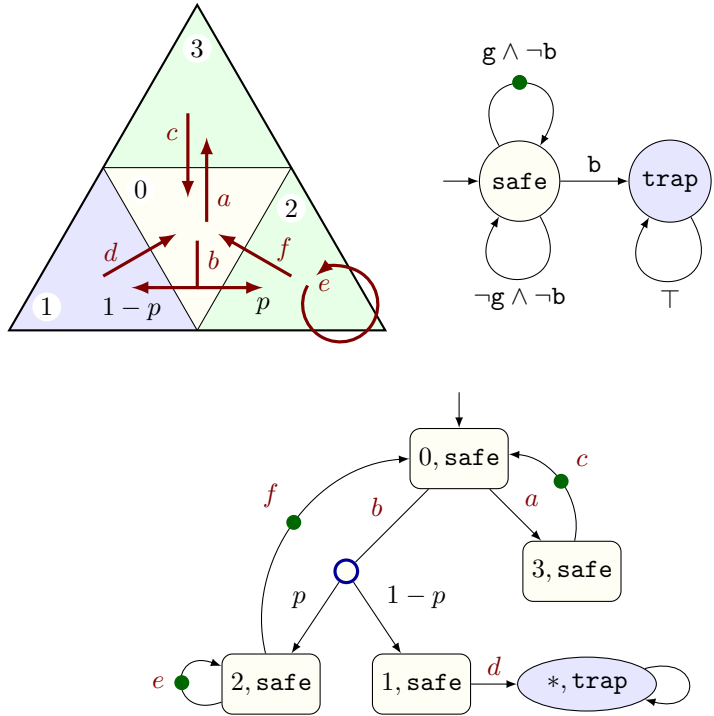

\begin{example}[DBW to Expected Average Reward Reduction] 
\label{ex:average-reward}
This counterexample demonstrates that even for deterministic B\"uchi objectives,
the problem of finding an optimal strategy satisfying an objective may
not be
reduced to the problem of finding an optimal average strategy. 
Consider the simple grid-world example of Figure~\ref{fig:grid-world}
with the specification
$\varphi = (\always \neg \mathtt{b}) \wedge (\always\eventually
\mathtt{g})$, where atomic proposition $\mathtt{b}$ (blue) labels Cell~1 and
atomic proposition $\mathtt{g}$ (green)  labels Cells~2 and~3.
Actions enabled in various cells and their probabilities are
depicted in the figure.

The strategy from Cell~$0$ that chooses Action~$a$ guarantees
satisfaction of $\varphi$ with probability $1$.  An automaton with
accepting transitions for $\varphi$ is shown in
Figure~\ref{fig:grid-world}; it is a DBW
(or equivalently a DRW with one pair $(B, G)$ and $B = \emptyset$).

The product MDP is shown at the bottom of Figure~\ref{fig:grid-world}.
All states whose second component is $\mathtt{trap}$ have been merged.
Notice that there is no negative reward since the set $B$ is empty.
If reward is positive and equal for all accepting transitions, and $0$
for all other transitions, then when $p > 1/2$, the strategy that
maximizes expected average reward chooses Action~$b$ in the initial
state and Action~$e$ from State~$(2,\mathtt{safe})$.  Note that for
large values of $\lambda$ the optimal expected average reward
strategies are also optimal strategies for the $\lambda$-discounted
reward objective. However, these strategies are not optimal for
$\omega$-regular objectives.
\end{example}


Example~\ref{ex:two-pairs} shows that one cannot select a pair from a
Rabin acceptance condition ahead of time.  This problem can be avoided
by the use of B\"uchi acceptance conditions.  While DBWs are not
sufficiently expressive, SLDBWs express all $\omega$-regular
properties and are suitable for probabilistic model checking.  In the
next section, we show that they are also ``the ticket'' for model-free
reinforcement learning because they allow us to maximize the
probability of satisfying an $\omega$-regular specification by
solving a reachability probability problem that can be solved
efficiently by off-the-shelf RL algorithm.


\section{Learning from Omega-Regular Rewards}
\label{sec:reinforcement}

\def\atm{\mathcal{A}}

Let us fix in this section an MDP $\Mm$, an $\omega$-regular property
$\varphi$ and its corresponding SLDBW $\mathcal{A}$.

A Markov chain is an MDP whose set of actions is a singleton.
A {\em bottom strongly connected component} (BSCC) of a Markov chain is any of its end-components.
A BSCC is accepting if it contains an accepting transition and otherwise it is a rejecting BSCC.
For any MDP $\Mm$ and positional strategy $\sigma$, let $\Mm_\sigma$ be the Markov chain resulting from resolving the nondeterminism in 
$\Mm$ using $\sigma$. 

We start with recalling two key properties of SLDBWs:
\begin{theorem} \cite{Hahn15}
  \label{th:recall}
 (1) 
 The probability of satisfying $\varphi$ in $\Mm$ is at least as large as the probability of satisfying the B\"uchi condition in $\Mm \times \atm$, and
 (2) if $p$ is the optimal probability of satisfying $\varphi$ in $\Mm$, then there exists a positional strategy $\sigma$ such that the probability of satisfying the B\"uchi condition in $(\Mm \times \atm)_\sigma$ is $p$.
\end{theorem}
Similar observations can also be found in other work, e.g.\  \cite{Vardi85,Courco95,Sicker16b}.
While the proof of (2) is involved, (1) follows immediately from the fact that any run in $\Mm$ 
induced by an accepting run of $\Mm \times \atm$ satisfies $\varphi$ (while the converse does not necessarily hold);
this holds for all nondeterministic automata, regardless of the acceptance mechanism.

\begin{definition}
  \label{def:augmented-mdp}
  For any $\zeta \in ]0,1[$, 
  the \emph{augmented MDP} $\Mm^\zeta$ is an MDP
  obtained from $\Mm\!\times\!\atm$ by adding a new 
  state $t$ with a self-loop to the set of states of
  $\Mm\!\times\!\atm$, and by making $t$ a destination of each accepting transition
  $\tau$ of $\Mm\!\times\!\atm$ with probability $1-\zeta$. 
  The original probabilities of all other destinations 
  of an accepting transition $\tau$ are multiplied by $\zeta$.
\end{definition}

With a slight abuse of notation, if $\sigma$ is a strategy on the augmented MDP
$\Mm^\zeta$, we denote by $\sigma$ also the strategy on $\Mm \!\times\!\atm$ obtained by
removing $t$ from the domain of $\sigma$.

We let $p^\sigma_s(\zeta)$ denote the probability of reaching 
$t$ in 
$\Mm^\zeta_\sigma$ when starting at $s$.
Notice that we can encode this value as the expected average reward in the following rewardful MDP
$(\Mm^\zeta, \rho)$ where we set 
the reward function $\rho(t,a) = 1$ for all $a \in A$ and $\rho(s,a) = 0$
otherwise.
For any strategy $\sigma$, $p^\sigma_s(\zeta)$ and the reward of $\sigma$
from $s$ in $(\Mm^\zeta, \rho)$ are the same.

We also let $a^\sigma_s$ be the probability that a run that starts from $s$ in $(\Mm\times\atm)_\sigma$ is accepting.
We now show the following basic properties of these values.
\begin{lemma}
  \label{lem:basic_properties}
  If $\sigma$ is a positional strategy on $\Mm^\zeta$, then, for every state
  $s$ of $(\Mm \!\times\!\atm)_\sigma$, the following holds:
  \begin{enumerate}
   \item If $s$ is in a rejecting BSCC of $(\Mm \!\times\!\atm)_\sigma$, then $p^\sigma_s(\zeta) = 0$.
   \item If $s$ is in an accepting BSCC of $(\Mm\!\times\!\atm)_\sigma$, then $p^\sigma_s(\zeta) = 1$.
   \item $p^\sigma_s(\zeta) \geq a^\sigma_s$
   \item If $p^\sigma_s(\zeta) = 1$ then no rejecting BSCC is reachable from $s$ in 
   $(\Mm \!\times\!\atm)_\sigma$ and $a^\sigma_s = 1$.
  \end{enumerate}
\end{lemma}

\begin{proof}
  (1) holds as there are no accepting transition in a rejecting BSCC of $(\Mm \!\times\!\atm)_\sigma$, 
  and so $t$ cannot be reached when starting at $s$ in $\Mm^\zeta$.
  (2) holds because $t$ (with its self-loop) is the only BSCC reachable from $s$ in $\Mm^\zeta$.
  In other words, $t$ (with its self-loop) and the rejecting BSCCs of $(\Mm \!\times\!\atm)_\sigma$ are the only BSCCs in $\Mm^\zeta$.
  (3) then follows, because the same paths lead to a rejecting BSCCs in $(\Mm \!\times\!\atm)_\sigma$ and $\Mm^\zeta_\sigma$, while the probability of each such a path is no larger---and strictly smaller iff it contains an accepting transition---than in $\Mm^\zeta_\sigma$.
  (4) holds because, if $p^\sigma_s(\zeta) = 1$, then $t$ (with its self-loop) is the only BSCC reachable from $s$ in $\Mm^\zeta_\sigma$.
  Thus, there is no path to a rejecting BSCC in $\Mm^\zeta_\sigma$, and therefore no path to a rejecting BSCC in $(\Mm \!\times\!\atm)_\sigma$.
\end{proof}

\begin{lemma}
  \label{th:reachability-probability}
  Let $\sigma$ be a positional strategy on $\Mm^\zeta$.
  For every state $s$ of $\Mm^\zeta$, 
  we have that $\llim\zeta1 p^\sigma_s(\zeta) = a^\sigma_s$.
\end{lemma}

\begin{proof}
As shown in Lemma \ref{lem:basic_properties}(3) for all $\zeta$ we have $p^\sigma_s(\zeta) \geq a^\sigma_s$.
For a coarse approximation of their difference, we recall that $(\Mm \!\times\!\atm)_\sigma$ is a finite Markov chain.
The expected number of transitions taken before reaching a BSCC from $s$ in $(\Mm \!\times\!\atm)_\sigma$ is therefore a finite number.
Let us refer to the (no larger) expected number of accepting transitions before reaching a BSCC 
when starting at $s$ in $(\Mm \!\times\!\atm)_\sigma$ as $f^\sigma_s$.
We claim that $a^\sigma_s \geq p^\sigma_s(\zeta)-(1-\zeta) \cdot f^\sigma_s$.
This is because the probability of reaching a rejecting BSCC in $(\Mm \!\times\!\atm)_\sigma$ 
is at most the probability of reaching a rejecting BSCC in $\Mm^\zeta_\sigma$, which is at most $1-p^\sigma_s(\zeta)$,
plus the the probability of moving on to $t$ from a state that is not in any BSCC in $(\Mm \!\times\!\atm)_\sigma$, 
which we are going to show next is at most $f^\sigma_s \cdot (1-\zeta)$.

First, a proof by induction shows that
$(1-\zeta^k) \leq k (1-\zeta)$ for all $k \geq 0$.
Let $P^\sigma_s(\zeta, k)$ be the probability of generating a path from $s$ with $k$ accepting transitions
before $t$ or a node in some BSCC of $(\Mm \!\times\!\atm)_\sigma$ is reached in $\Mm^\zeta_\sigma$.
The probability of seeing $k$ accepting transitions and not reaching $t$ is at least $\zeta^k$.
Therefore, the probability of moving to $t$ from a state not in any BSCC is at most
$\sum_k P^\sigma_s(\zeta, k) (1{-}\zeta^k) \leq \sum_k P^\sigma_s(\zeta, k) k \cdot
(1{-}\zeta) \leq f^\sigma_s \cdot (1{-}\zeta)$.
\end{proof}

This provides us with our main theorem.
\begin{theorem}
  \label{th:main}
  There exists a threshold $\zeta' \in ]0,1[$ such that, for all
  $\zeta > \zeta'$ and every state $s$, any strategy $\sigma$ that
  maximizes $p^{\sigma}_s(\zeta)$ in $\Mm^\zeta$ is (1) an optimal
  strategy in $\Mm\times\atm$ from $s$ and (2) induces an optimal for
  the original MDP $\Mm$ from $s$ with objective $\varphi$.
\end{theorem}

\begin{proof}
We use the fact that it suffices to study positional strategies and there are only finitely many of them.
Let $\sigma_1$ be an optimal strategy of $\Mm \!\times\!\atm$, 
and let $\sigma_2$ be a strategy
that has the highest likelihood of creating an accepting run among 
all non-optimal memoryless ones.
(If $\sigma_2$ does not exist, then all strategies are equally good, and it does not matter which one is chosen.)
Let $\delta = a^{\sigma_1}_s - a^{\sigma_2}_s$.

Let $f_{\max} = \max_\sigma \max_s f^\sigma_s$ where $\sigma$ ranges over positional strategies only and
$f^\sigma_s$ is defined as in Lemma \ref{th:reachability-probability}.
We claim that it suffices to pick $\zeta'\in ]0,1[$ such that $(1-\zeta')\cdot f_{\max} < \delta$.

Suppose that $\sigma$ is a positional strategy that is optimal in $\Mm^\zeta$ for $\zeta > \zeta'$,
but is not optimal in $\Mm \!\times\!\atm$. We then have
$a^{\sigma}_s \leq p^\sigma_s(\zeta) \leq a^\sigma_s + (1-\zeta)f^\sigma_s < a^\sigma_s + \delta \leq a^{\sigma_1}_s \leq p^{\sigma_1}_s(\zeta)$,
where these inequalities follow, respectively, from: Lemma \ref{lem:basic_properties}(3),
the proof of Lemma \ref{th:reachability-probability},
the definition of $\zeta'$, the assumption that $\sigma$ is not optimal
and the definition of $\delta$, and the last one from Lemma \ref{lem:basic_properties}(3).
This shows that $p^\sigma_s(\zeta) < p^{\sigma_1}_s(\zeta)$, i.e., $\sigma$ is not optimal in $\Mm^\zeta$; a contradiction.
Therefore, any positional strategy that is optimal in $\Mm^\zeta$ for $\zeta > \zeta'$
is also optimal in $\Mm \!\times\!\atm$.


Now, suppose that $\sigma$ is a positional strategy that is optimal in $\Mm \!\times\!\atm$.
Theorem \ref{th:recall}(1) implies then that
the probability of satisfying $\varphi$ in $\Mm$ when starting at $s$ is at least $a^\sigma_s$.
At the same time, if there was a strategy for which 
the probability of satisfying $\varphi$ in $\Mm$ is $ > a^\sigma_s$, 
then Theorem \ref{th:recall}(2) would guarantee the existence of strategy $\sigma'$
for which $a^{\sigma'}_s > a^\sigma_s$; a contradiction with the assumption that $\sigma$ is optimal.
Therefore any positional strategy that is optimal in $\Mm \!\times\!\atm$
induces an optimal strategy in $\Mm$ with objective $\varphi$.
\end{proof}

Note that, due to Lemma \ref{lem:basic_properties}(4), every
$\zeta \in ]0,1[$ works when our objective is to satisfy $\varphi$
with probability $1$.

\section{Experimental Results}
\label{sec:experiment}

We implemented the construction described in the previous sections in
a tool named \toolname \cite{Eliot39}, which reads MDPs described in
the PRISM language \cite{kwiatk11}, and $\omega$-regular automata
written in the HOA format \cite{Babiak15,hoaf}.  \toolname builds the
product $\Mm^\zeta$, provides an interface for RL algorithms akin to
that of \cite{Brockm16} and supports probabilistic model checking.
Our algorithm computes, for each pair $(s,a)$ of state and action, the
maximum probability of satisfying the given objective after choosing
action $a$ from state $s$ by using off-the-shelf, temporal difference
algorithms.  Not all actions with maximum probability are part of
positional optimal strategies---consider a product MDP with one state
and two actions, $a$ and $b$, such that $a$ enables an accepting
self-loop, and $b$ enables a non-accepting one: both state/action
pairs are assigned probability $1$.  Since the probability values
alone do not identify a pure optimal strategy, \toolname computes an
optimal mixed strategy, uniformly choosing all maximum probability
actions from a state.


The MDPs on which we tested our algorithms are listed in
Table~\ref{tab:experiment}.  For each model, the number of states in
the MDP and automaton is given, together with the maximum probability
of satisfaction of the objective, as computed by the RL algorithm and
confirmed by the model checker (which has full access to the MDP).
Models \texttt{twoPairs} and \texttt{riskReward} are from
Examples~\ref{ex:two-pairs} and~\ref{ex:average-reward}, respectively.
Models \texttt{grid5x5} and \texttt{trafficNtk} are from
\cite{Sadigh14}.  The three ``windy'' MDPs are taken from
\cite{Sutton18}.  The ``frozen'' examples are from \cite{OpenaiGym}.
Some $\omega$-regular objectives are simple reachability requirements
(e.g., \texttt{frozenSmall} and \texttt{frozenLarge}).  The objective for
\texttt{anothergrid} is to collect three types of coupons, while
incurring at most one of two types of penalties.  The objective for
\texttt{slalom} is given by the LTL formula
$\always(p \rightarrow \nextt \always \neg q) \wedge \always(q
\rightarrow \nextt \always \neg p)$.


\begin{table}
  \centering
  \begin{tabular}[c]{l|ccccc}
    Name & states & automaton & probability \\\hline
    \texttt{twoPairs} & 4 & 4 & 1 \\
    \texttt{riskReward} & 4 & 2 & 1 \\
    \texttt{deferred} & 41 & 1 & 1 \\
    \texttt{dpenny} & 52 & 2 & 0.5 \\
    \texttt{anothergrid} & 36 & 25 & 1 \\
    \texttt{frozenSmall} & 16 & 3 & 0.823 \\
    \texttt{frozenLarge} & 64 & 3 & 1 \\
    \texttt{grid5x5} & 25 & 3 & 1  \\
    \texttt{windy} & 63 & 1 & 1  \\
    \texttt{windyKing} & 70 & 1 & 1 \\
    \texttt{windyStoch} & 70 & 1 & 1  \\
    \texttt{slalom} & 38 & 5 & 1 \\
    \texttt{devious} & 11 & 1 & 1  \\
    \texttt{mdp2} & 6 & 3 & 1 \\
    \texttt{mdp3} & 6 & 3 & 1 \\
    \texttt{morris} & 4 & 4 & 1 \\
    \texttt{arbiter2} & 32 & 3 & 1 \\
    \texttt{corridor} & 8 & 2 & 1 \\
    \texttt{knuthYao} & 13 & 3 & 1 \\
    \texttt{strat1} & 3 & 4 & 1 \\
    \texttt{trafficNtk} & 122 & 9 & 1
  \end{tabular}
  \caption{Experimental results.}
  \label{tab:experiment}
\end{table}

Figure~\ref{fig:deferred} illustrates how increasing the parameter
$\zeta$ makes the RL algorithm less sensitive to the presence of
transient accepting transitions.  Model \texttt{deferred} consists of
two chains of states: one, which the agent choses with action $a$, has
accepting transitions throughout, but leads to an end component that
is not accepting.  The other chain, selected with action $b$, leads to an
accepting end component, but has no other accepting transitions.  There
are no other decisions in the model; hence only two strategies are
possible, which we denote by $a$ and $b$, depending on the action
chosen.

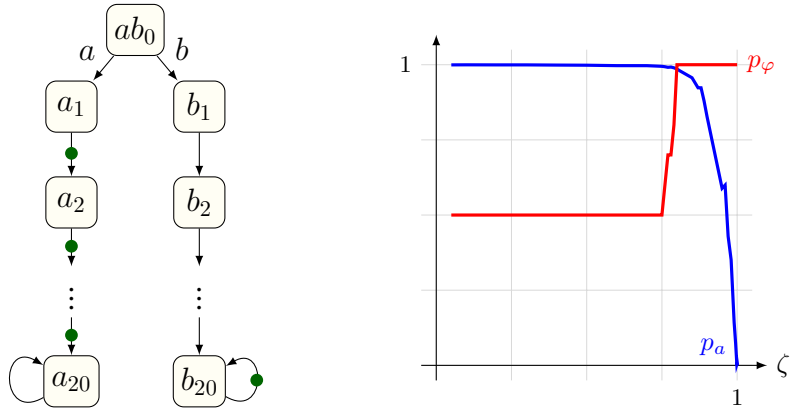
\begin{figure}
  \centering
  \begin{tikzpicture}
    \begin{scope}[scale=0.85,yshift=5.25cm,transform shape,
      every node/.style={font=\Large}]
      \node[box state,fill=safecellcolor] (AB0) {$ab_0$};
      \node[box state,fill=safecellcolor] (A1) [below left=1.2cm and 1cm
      of AB0] {$a_1$};
      \node[box state,fill=safecellcolor] (B1) [below right=1.2cm and
      1cm of AB0] {$b_1$};
      \node[box state,fill=safecellcolor] (A2) [below=1.5cm of A1]
      {$a_2$};
      \node[box state,fill=safecellcolor] (B2) [below=1.5cm of B1]
      {$b_2$};
      \node (DotsA) [below=1.4 of A2] {$\vdots$};
      \node (DotsB) [below=1.4 of B2] {$\vdots$};
      \node[box state,fill=safecellcolor] (A20) [below=1.4cm of DotsA]
      {$a_{20}$};
      \node[box state,fill=safecellcolor] (B20) [below=1.4cm of DotsB]
      {$b_{20}$};
      \path[->]
      (AB0) edge[swap] node {$a$} (A1)
      edge node {$b$} (B1)
      (A1) edge node [accepting dot] {} (A2)
      (B1) edge (B2)
      (A2) edge node [accepting dot] {} (DotsA)
      (B2) edge (DotsB)
      (DotsA) edge node [accepting dot] {} (A20)
      (DotsB) edge (B20)
      (A20) edge[loop left,looseness=4,out=210,in=150] (A20)
      (B20) edge[loop right,looseness=4,out=330,in=30] node [accepting dot] {} (B20);
    \end{scope}
    \begin{scope}[xshift=4cm,scale=4]
      \draw[ultra thin, color=gray!30, step=0.25] (-0.05,-0.05) grid (1.05,1.05);
      \draw[->] (-0.05,0) -- (1.1,0) node[right] {$\zeta$};
      \draw[->] (0,-0.05) -- (0,1.1);
      \draw[blue, very thick] plot coordinates {(0.05,0.999454)
        (0.1,0.999645) (0.2,0.999403) (0.3,0.999431) (0.4,0.998791)
        (0.5,0.998318) (0.6,0.996904) (0.7,0.996797) (0.75,0.99465)
        (0.76,0.99327) (0.77,0.99133) (0.78,0.9922) (0.79,0.98969)
        (0.8,0.9857) (0.81,0.97941) (0.85,0.956589) (0.87,0.923285)
        (0.88,0.92323) (0.89,0.87954) (0.9,0.827172) (0.95,0.588913)
        (0.96,0.599504) (0.97,0.429409) (0.98,0.350915) (0.99,0.142736)
        (0.995,0.072589) (0.999,0.003537) (0.9995,0.005199) (0.9999,0.0)
      } node[above left] {$p_a$};
      \draw[red, very thick] plot coordinates {(0.05,0.5)
        (0.1,0.5) (0.2,0.5) (0.3,0.5) (0.4,0.5)
        (0.5,0.5) (0.6,0.5) (0.7,0.5) (0.75,0.5)
        (0.76,0.6) (0.77,0.7) (0.78,0.7) (0.79,0.8)
        (0.8,1) (0.81,1) (0.85,1) (0.87,1)
        (0.88,1) (0.89,1) (0.9,1) (0.95,1)
        (0.96,1) (0.97,1) (0.98,1) (0.99,1)
        (0.995,1) (0.999,1) (0.9995,1) (0.9999,1)
      } node[right] {$p_\varphi$};
      \node[left] (Y1) at (-0.05,1) {\small$1$}; 
      \node[below] (X1) at (1,-0.05) {\small$1$}; 
    \end{scope}
  \end{tikzpicture}
  \caption{Model \texttt{deferred} and effect of $\zeta$ on it.}
  \label{fig:deferred}
\end{figure}

The curve labeled $p_a$ in Figure~\ref{fig:deferred} gives the
probability of satisfaction under strategy $a$ of the MDP's objective
as a function of $\zeta$ as computed by Q-learning.  The number of
episodes is kept fixed at $20,000$ and each episode has length
$80$.  Each data point is the average of five experiments for the same
value of $\zeta$.

For values of $\zeta$ close to $0$, the chance is high that the sink
is reached directly from a transient state.  Consequently, Q-learning
considers strategies $a$ and $b$ equally good.  For this reason, the
probability of satisfaction of the objective, $p_\varphi$, according
to the strategy that mixes $a$ and $b$, is computed by \toolname's
model checker as $0.5$.  As $\zeta$ approaches $1$, the importance of
transient accepting transitions decreases, until the probability
computed for strategy $a$ is no longer considered to be approximately
the same as the probability of strategy $b$.  When that happens,
$p_\varphi$ abruptly goes from $0.5$ to its true value of $1$, because
the pure strategy $b$ is selected.  The value of $p_a$ continues to
decline for larger values of $\zeta$ until it reaches its true value
of $0$ for $\zeta = 0.9999$.  Probability $p_b$, not shown in the
graph, is $1$ throughout.

The change in value of $p_\varphi$ does not contradict
Theorem~\ref{th:main}, which says that $p_b=1 > p_a$ for all
values of $\zeta$.  In practice a high value of $\zeta$ may be needed
to reliably distinguish between transient and recurrent accepting
transitions in numerical computation.  Besides, Theorem \ref{th:main}
suggests that even in the almost-sure case there is a meaningful path
to the target strategy where the likelihood of satisfying $\varphi$
can be expected to grow.  This is important, as it comes with the
promise of a generally increasing quality of intermediate strategies.


\end{document}